\documentclass[smallextended]{svjour3-arxiv}     
\smartqed  
\usepackage{graphicx}
\usepackage{wrapfig}
\newcommand\eat[1]{}
\usepackage{tikz}
\usepackage{booktabs}  
\usepackage{lineno}
\usepackage{blindtext}

\usepackage{slashbox}
  \journalname{}
 
\usepackage{array,xspace,multirow,hhline,graphicx,xcolor,tikz,colortbl,tabularx,amsmath,amssymb,amsfonts}
\usepackage[round]{natbib}
\usetikzlibrary{shapes}
\usepackage{algorithm,algorithmic}
\usepackage{mathrsfs}
\usepackage{enumitem}
\setenumerate[1]{label=(\emph{\roman{*}}),ref=(\emph{\roman{*}}),leftmargin=*}

\usepackage{txfonts}
\usepackage{eucal} 

\newcommand{\eg}{e.g.,\xspace}

\usepackage{tikz}

\newlength{\wordlength}
\newcommand{\wordbox}[3][c]{\settowidth{\wordlength}{#3}\makebox[\wordlength][#1]{#2}}

\newcommand{\pref}{\succsim\xspace}
\newcommand{\Pref}[1][]{
	\ifthenelse{\equal{#1}{}}{\mathrel \succsim}{\mathop{\succsim_{#1}}}
}                                          
\newcommand{\sPref}[1][]{                  
	\ifthenelse{\equal{#1}{}}{\mathrel \succ}{\mathop{\succ_{#1}}}
}                                          
\newcommand{\Indiff}[1][]{                 
	\ifthenelse{\equal{#1}{}}{\mathrel \sim}{\mathop{\sim_{#1}}}
}
\newcommand{\prefset}[1][]{\ifthenelse{\equal{#1}{}}{\mathcal{\succsim}}{\mathcal{\succsim}_{#1}}}

\newcommand{\midd}{\mathbin{: }}

\usepackage{mathtools}

\DeclarePairedDelimiter\floor{\lfloor}{\rfloor}

\let\enumtemp=\enumerate
\def\enumerate{\enumtemp\itemsep 1pt}
\let\itemtemp=\itemize
\def\itemize{\itemtemp\itemsep 1pt}

\newcommand{\Omit}[1]{}

\begin{document}

		\title{A Generalization of the AL method \\for Fair Allocation of Indivisible Objects}


	\author{Haris Aziz}


	\institute{%
	  Haris Aziz \at
	  NICTA and UNSW
	  2033 Sydney, Australia \\
	  Tel.: +61-2-8306\,0490 \\
	  Fax: +61-2-8306\,0405 \\
	  \email{haris.aziz@nicta.com.au}
	}


\maketitle

	\begin{abstract}
		We consider the assignment problem in which agents express ordinal preferences over $m$ objects and the objects are allocated to the agents based on the preferences.
In a recent paper, Brams, Kilgour, and Klamler (2014) presented the AL method to compute an envy-free assignment for two agents. The AL method crucially depends on the assumption that agents have strict preferences over objects. We generalize the AL method to the case where agents may express indifferences and prove the axiomatic properties satisfied by the algorithm. 
As a result of the generalization, we also get a $O(m)$ speedup on previous algorithms to check whether a complete envy-free assignment exists or not.  Finally, we show that unless P=NP, there can be no polynomial-time extension of GAL to the case of arbitrary number of agents.
	\end{abstract}
	


%

	\section{Introduction}
	
Fair allocation of resources is one of the most critical issues for society. A basic, yet widely applicable, problem in computer science and economics is to allocate discrete objects to agents given the ordinal preferences of the agents over the objects. The setting is referred to as the \emph{assignment problem} or the \emph{house allocation problem}~\citep[see, \eg][]{ACMM05a,AGMW14a,BEL10a,BrKa04a,BEF03a,BrFi00a,BKK12a,DeHi88a,Gard73b,Manl13a,Wils77a,Youn95b}. 
		In this setting, there is a set of agents $N=\{1,\ldots, n\}$, a set of objects $O=\{o_1,\ldots, o_{m}\}$ with each agent $i\in N$ expressing ordinal preferences $\pref_i$ over $O$. Each object is assumed to be acceptable to the agents. The goal is to allocate the objects among the agents in a fair or optimal manner without allowing transfer of money. 
		The model is applicable to many resource allocation or fair division settings where the objects may be public houses, school seats, course enrolments, kidneys for transplant, car park spaces, chores, joint assets of a divorcing couple, or time slots in schedules. 
		
For the assignment problem, the case of two agents is especially central. Many disputes are between two parties and may require division of common resources. Divorce proceedings is one of the settings in which common assets need to be divided among the two parties. Other examples in history include partition of countries which results in the need to divide common assets.

When objects are allocated among agents, it is desirable that they are allocated in a fair and efficient manner. 
For fairness, one of the most established concepts is envy-freeness. A formal study of envy-freeness in microeconomics can be traced back to the work of \citet{Fole67a}.
Envy-freeness requires that each agent should prefer its allocation over other agents' allocations. Envy-freeness can be trivially satisfied by not giving any objects to any agents. However, if we insist 
that the assignment should be \emph{complete}, i.e., it allocated all the objects to the agents, 
no assignment may be envy-free as is the case in which there is only one object and the agent who does not get any object is envious. 
The most established notion of efficiency is Pareto optimality which requires that there should be no other allocation which each agent weakly prefers and at least one agent strictly prefers. Pareto optimality has been termed the \emph{``single most important tool of normative economic analysis''}~\citep{Moul03a}. 
An assignment which gives no objects to agents is clearly not Pareto optimal.


In view of the importance of the two-agent setting, and the fundamental goals of envy-freeness and Pareto optimality, \citet{BKK14a} presented an elegant algorithm called AL for the case of two agents that computes a maximal assignment that is envy-free as well as \emph{locally Pareto optimal} (Pareto optimal for the set of allocated objects).\footnote{The notion of envy-freeness that they use is equivalent to SD (stochastic dominance) envy-freeness~\citep{AGMW14a} and necessary envy-freeness~\citep{BEL10a}.} 
Since there may not exist a Pareto optimal and envy-free assignment, \citet{BKK14a} relax the requirement of Pareto optimality to local Pareto optimality. 
The algorithm has received attention in the literature~\citep[see \eg][]{Bram14a,Bram14b,BCM15a,DGG+14a,PrWa14a}.
The desirable aspect of AL is that is returns a locally Pareto optimal and a maximal envy-free assignment. 
By maximal, we mean that unallocated objects cannot be additionally given to the agents' partial allocations without compromising envy-freeness. \citet{BKK14a} also claim that AL returns a complete envy-free assignment if there exists a complete envy-free assignment.
One possible limitation of the AL method is that it assumes that agents have strict preferences over objects. We present a generalization of the AL method in which agents may express indifferences among agents.

Indifferences in preferences are not only a natural relaxation but are also a practical reality in many cases. 
For example, if there are multiple copies of the same object with the same characteristics, then an agent is invariably indifferent among all such copies.
Indifferences can lead to various challenges. The complexity of solution concepts in the presence of indifferences can be considerably more than in the case of strict preferences. A famous example is that of roommate markets for which the problem of finding a stable matching is polynomial-time solvable for strict preferences but NP-complete for weak orders~\citep{Ronn90a}. Similarly, a number of fairness concepts are harder to compute when weak orders are allowed~\citep{AGMW14a}. In view of this, effort has been taken to generalize algorithms and rules for the case of indifferences in 
voting~\citep[see \eg][]{ABBH12a,CHP13a}, housing markets~\citep[see \eg][]{AzKe12a,SeSa13a}, coalition formation~\citep{ABH11c}, and various matching market models~\citep{IwMi08a,Manl13a,Scot05a}. 
The main contribution of this paper is a generalization of AL  which we refer to as GAL for the case in which agents may express indifferences. 
The main result of the paper is as follows. 

		\begin{theorem}
For two agents, GAL  returns in time $O(m^2)$ a maximal envy-free and locally Pareto optimal assignment even if agents express weak preferences. If a complete envy-free assignment exists, GAL computes a complete envy-free assignment.
Moreover, there exists no other assignment that Pareto dominates it and is still envy-free. 
			\end{theorem}
	
	Previously, ~\citet{BEL10a} and \citet{AGMW14a} presented $O(m^3)$ time algorithms to check whether a \emph{complete} envy-free assignment for two agents exists or not. The algorithms require solving network flow or maximum matching problems. As a corollary of GAL, we obtain a simple $O(m^2)$ algorithm to check whether there exists a complete assignment that is EF.

	The critical reader may ask whether GAL can be generalized to handle an arbitrary number of agents.
We show that unless complexity classes P and NP coincide~\citep{Fort13a}, there exists \emph{no} polynomial-time algorithm for an arbitrary number of agents that satisfies the same properties as GAL.
	
	\section{Related work}


	Computation of fair discrete assignments has been intensely studied in the last decade. In many of the papers considered, agents express cardinal utilities for the objects and the goal is to compute fair assignments~\citep[see \eg][]{LMMS04a,PrWa14a}. 
We consider the setting in which agents only express ordinal preferences over objects~\citep{AGMW14a,BEL10a,BrKa04a,BEF03a,BrFi00a,PrWo12a} which are less demanding to elicit. 
	
	When agents express preferences over objects and we need to reason about preferences over allocations, there are different ways one can define envy-freeness such as possible envy-freeness and weak SD envy-freeness~\citep{AGMW14a}. In this paper we will use the strongest known reasonable notion of envy-freeness. The notion is equivalent to necessary envy-freeness in~\citep{BEL10a}, SD-envy-freeness in ~\citep{AGMW14a}, EF notion used in ~\citep{BKK14a} and pair-wise envy-freeness in ~\citep{BKK14b}. We will refer to the notion simply as EF just like \citet{BKK14a} do.
\citet{AGMW14a} and \citet{BEL10a} presented $O(m^3)$ algorithms to check whether there exists an EF assignment. We show that there exists a simple $O(m^2)$ algorithm for the problem even if agents express weak preferences. 		
	

There are other papers~\citep{Aziz14a,CDE+06a,KBKZ09a} in fair division in which agents  explicitly express ordinal preferences over sets of objects rather than simply expressing preferences over objects~\citep{Aziz14a,CDE+06a,KBKZ09a}. For these more expressive models, the computational complexity of computing fair assignments is either even higher~\citep{CDE+06a,KBKZ09a} or representing preferences requires exponential space~\citep{Aziz14a,BKK12a}. In this paper, we 
restrict agents to simply express ordinal preferences over objects.

	\section{Preliminaries}

	An assignment problem is a triple $(N,O,\pref)$ such that $N=\{1,\ldots, n\}$ is a set of agents, $O=\{o_1,\ldots, o_m\}$ is a set of objects, and the preference profile $\pref=(\pref_1,\ldots, \pref_n)$ specifies for each agent $i$ its preference $\pref_i$ over $O$.
	Agents may be indifferent among objects. We will denote by $\succ_i$ the strict part and by $\sim_i$ the indifference part of the relation $\pref_i$.
	We denote $\pref_i: E_i^1,\dots,E_i^{k_i}$ for each agent $i$ with equivalence classes in decreasing order of preferences.
	Thus, each set $E_i^j$ is a maximal equivalence class of objects among which agent $i$ is indifferent, and $k_i$ is the number of equivalence classes of agent $i$. If an equivalence class is a singleton $\{o\}$, we list the object $o$ in the list without the curly brackets. A preference profile consists of \emph{dichotomous preferences} each agent has at most two equivalence classes. A preference profile consists of \emph{strict preferences} each agent has strict preferences over the objects.
	
\begin{definition}[Assignment]	
An assignment $p=(p(1),\ldots, p(n))$ specifies the \emph{allocation} of objects $p(i)$ to each $i\in N$ such that $p(i)\subseteq O$ and $p(i)\cap p(j)=\emptyset$ for all $i\neq j$. An assignment is \emph{complete} if $o\in \bigcup_{i\in N}p(i)$ for all $o\in O$.
\end{definition}

	%

We define the \emph{SD (stochastic dominance)} relation. An agent \emph{SD-prefers} one allocation over another if for each object, the former allocation gives the agent as much probability of getting at least preferred an object as the latter allocation.
\begin{definition}[SD (stochastic dominance)]
Given two assignments $p$ and $q$, $p(i) \succsim_i^{SD} q(i)$, i.e., agent $i$ \emph{SD~prefers} allocation $p(i)$ to allocation $q(i)$ if for each $o\in O$,
		\[ \left|\{o'\midd o'\pref_i o\} \cap p(i)\right| \geq \left|\{o'\midd o'\pref_i o\} \cap q(i)\right|.\]
Agent $i$ \emph{strictly SD prefers} $p(i)$ to $q(i)$: $p(i)\succ_i^{SD} q(i)$ if $p(i) \succsim_i^{SD} q(i)$ and $\neg [q(i) \succsim_i^{SD} p(i)]$.
\end{definition}

 Although each agent $i$ expresses ordinal preferences over objects, he could have a private cardinal utility $u_i$ consistent with $\pref_i$: $u_i(o)\geq u_i(o') \text{ if and only if } o\pref_i o'.$

\begin{definition}[SD-efficiency]
An assignment $p$ is \emph{SD-efficient} if  there exists no other assignment $q$ such that $q(i)\pref_i^{SD} p(i)$ for all $i\in N$ and $q(i)\succ_i^{SD} p(i)$ for some $i\in N$. 
\end{definition}
SD-efficiency is equivalent to Pareto optimality for discrete assignments as defined by \citep{BKK14a}. Hence we will refer to SD-efficiency as Pareto optimality and SD-domination as Pareto domination.

\begin{definition}[Locally Pareto optimal]
An assignment $p$ is is \emph{LPO (locally Pareto optimal)} if  there exists no other assignment $q$ such that $\bigcup_{i\in N}p(i)=\bigcup_{i\in N}q(i)$ and $q(i)\pref_i^{SD} p(i)$ for all $i\in N$ and $q(i)\succ_i^{SD} p(i)$ for some $i\in N$. 
\end{definition}

	\begin{definition}[SD envy-freeness]
An assignment $p$ satisfies \emph{SD envy-freeness} or is \emph{SD envy-free} if each agent  SD prefers its allocation to that of any other agent: 
   	\[p(i) \pref_i^{SD} p(j) \text{ for all } i,j\in N.\]
	\end{definition}
	From the definition it is easy to see that a necessary condition for SD envy-freeness is that each agent gets the same number of objects.

	\citet{BKK14a} defined EF as follows.\footnote{We use the same definition as in \citep{BKK14a} except that we use weakly prefers rather than strictly prefers since we are considering weak preferences.} 
	
	\begin{definition}[EF (envy-freeness)]
	An allocation $p$  is \emph{EF (envy-free)} if for all $i,j\in N$ $|p(i)|=|p(j)|$ and there exists an injection $f_{ij}:p(i)\rightarrow p(j)$  and an injection $f_{ji}:p(j)\rightarrow p(i)$ such that for each object $o\in p(i)$, $i$ (weakly) prefers $o$ to $f_{ij}(o)$ and for each object $o\in p(j)$, $j$ (weakly) prefers $o$ to $f_{ji}(o)$. 
	\end{definition}
	
	Then, by using a similar argument as \citep[Lemma 1, ][]{BKK14a}, we can show that EF is equivalent to SD envy-freeness. We detail the argument for the sake of completeness and to formally extend Lemma~1\citep{BKK14a} to the case of indifferences.

	\begin{lemma}
	EF is equivalent to SD envy-freeness. 
	\end{lemma}
	\begin{proof}
		We first show that EF implies SD envy-freeness. Suppose $p$ satisfies EF and take any object $o\in O$. 
	Suppose that there is an object $o'\in p(i)$ such that $f_i(o')\succsim_i o$. By the definition of $f_i$, we know that $o'\pref_i f_i(o')$. Since  $f_i(o')\succsim_i o$, we get that $o'\pref_i f_i(o')\succsim_i o$. Hence, 	

			\[ \left|\{o'\midd o'\pref_i o\} \cap p(i)\right| \geq \left|\{o'\midd o'\pref_i o\} \cap p(-i)\right|.\]


	We now show that SD envy-freeness implies EF.
	Suppose that assignment $p$ does not satisfy EF with the EF condition violated for agent $i$. 
	Consider a bipartite graph $G=(p(-i)\cup p(i),E)$ where $\{o,o'\}\in E$ if $o\in p(-i)$, $o'\in p(i)$, and $o'\succsim_i o$. Since $p$ does not satisfy EF for $i$,  $G$ does not admit a perfect matching. By Hall's theorem, there exists set $O'\subseteq p(-i)$ such that $|N(O')|<|O'|$ where $N$ is the neighbourhood of $O'$ in the graph $G$. Consider an object $o\in \min_{\pref_i}(O')$. Since, $|N(O')|<|O'|$,  this it implies that 
			\[ \left|\{o'\midd o'\pref_i o\} \cap p(i)\right| < \left|\{o'\midd o'\pref_i o\} \cap p(-i)\right|.\]

	But then $p$ does not satisfy SD envy-freeness.
	\qed
	\end{proof}

	\begin{lemma}\label{lemma:efcheck}
If the number agents is constant, it can be checked in $O(m)$ time whether a given assignment is EF or not.
		\end{lemma}
		\begin{proof}
We show that it can be checked in $O(m)$ time whether a given assignment for a constant number agents is SD envy-free or not.
		We first show that an SD comparison between any two allocations can be made in $O(m)$ time. 
		Let us say that we want to check whether $p(i)\succsim_i^{SD} p(-i)$ where $-i$ is some agent other than $i$.  Without loss of generality, assume that $i$'s preferences are a coarsening of linear order $o_1,\ldots, o_m$.
		\begin{itemize}
			\item We construct in $O(m)$ a vector $x(p(i))=(x_1,\ldots, x_m)$ where $x_i=1$ if $o_i\in p(i)$ and $x_i=0$ otherwise. 
Using $x(p(i))$ we construct in $O(k_i)$ time a vector $s'(p(i))=(s_1',\ldots, s_{k_i}')$ where $s_j'=|\{E_i^j\}\cap p(i)|$. Using $s'(p(i))$ we construct in $O(k_i)$ time a vector $s(p(i))=(s_1,\ldots, s_{k_i})$ where $s_j=\sum_{\ell=1}^js_j'$.
\item In a similar way, we construct 
in $O(m)$ a vector $y(p(-i))=(y_1,\ldots, y_m)$ where $y_i=1$ if $o_i\in p(-i)$ and $x_i=0$ otherwise. Using $y(p(i))$ we construct in $O(k_i)$ time a vector $t'(p(-i))=(t_1',\ldots, t_{k_i}')$ where $t_j=|\{E_i^j\}\cap p(-i)|$. Using $t'(p(-i))$ we construct in $O(k_i)$ time a vector $t(p(i))=(t_1,\ldots, t_{k_i})$ where $t_j=\sum_{\ell=1}^jt_{\ell}'$.
		\end{itemize}

Now $p(i)\succsim_i^{SD} p(-i)$ iff $s_j\geq t_j$ for all $j\in \{1,\ldots, k_i\}$. This again takes time $O(k_i)$. Hence an SD comparison between allocation takes time $O(m)+4O(k_i)=O(m)$.

In order to test EF, we need to make $n(n-1)$ comparisons which is constant if $n$ is constant. Hence testing EF of an assignment for constant number of agent takes time $O(m)$.\qed
\end{proof}

If the number of agents is not constant, then the time complexity is $O(n^2m)$. In the paper, we will assume that $n=2$ i.e., there are two agents. If we refer to some agent as $i\in \{1,2\}$, then we will refer to the other agent as $-i$. Even for more than two agents, we may refer to $-i$ as some agent other than $i\in N$.

	Finally, define \emph{maximal envy-freeness}.
	\begin{definition}[Maximally envy-free assignment]
	We say that a partial assignment $p$ is \emph{maximally envy-free} if it is envy-free and
	there exists no assignment $q$ such that $q(i)\supseteq p(i)$ for all $i\in N$, $q(i)\supset p(i)$ for some $i\in N$, and $q$ is envy-free.
\end{definition}

	%
	%
	%
	%
	%
	%

	%
	%

	%
	%
	%

	\section{GAL --- Generalized AL}


	Before we delve into GAL, we first informally describe a simplified version of AL that still satisfies the properties of AL as described in \citep{BKK14a}. Agents have strict preferences and in each round they pick one object each. The algorithm repeats the following until all objects have been allocated to agent $1$, $2$, or contested pile $C$. We will refer to an object as \emph{unallocated} if it has not been allocated to $1$ or $2$ or placed in $C$.
If the most preferred unallocated object of the agents is not the same, each agent picks its most preferred object. Otherwise, if the most preferred unallocated object $o$ coincides, then we check whether we can give it to agent $1$. If $o$ is given to agent $1$ and the next most preferred unallocated object is given to agent $2$ and the partial assignment satisfies EF, then we allow such an allocation in the round. If not, we check in the same way whether we can give it to agent $2$.\footnote{This feasibility check is phrased in a different way in the original description of AL but is equivalent to checking for EF.} If $o$ cannot be given to either of the two, we put it in $C$. 
	
	%
	%
	%
	%
	%

	The general idea of GAL is as follows. Since the preferences of the two agents are weak orders, we first construct unique linear orders called \emph{priority orders} based on the preferences. Although, the comparisons to check the feasibility of EF assignments are still done with respect to the original preferences, the constructed linear orders help identify which unique object should each agent try to get first. 
The priority orders are refinements of the preferences where, if an agent is indifferent between two objects, it has higher priority for the object less preferred by the other agent. If both agents are indifferent among two objects, then agent 1 has higher priority for the object with the lower index and agent 2 has higher priority for the object with the higher index. After suitably constructing the linear orders, $>_1$ and $>_2$, agents try to take the highest priority. If agents have a different highest priority object, they take their highest priority objects. Otherwise there is a conflict so we must try to give one of the agents the highest priority object and give the other agent the second highest priority object according to the priority list if it does not violate EF.\footnote{The view of EF as being defined with respect to the SD relation makes it easy to argue for a maximal EF assignment.} 
If this cannot be done, we send the contested object to $C$, the so called \emph{contested pile}. A key idea behind GAL  is that if an object $o^*$ is sent to the contested pile, then it cannot be the case that $o^*$ along with 
some subsequent less preferred objects are allocated to agents and EF is not violated.
The algorithm is formally defined as Algorithm~\ref{algo:GAL}. Note there is an asymmetry in the algorithm in that agent one is considered first to get object $o^*$ in Step~\ref{step:agent1}. One can consider any of the two agents first or even toss a coin to select one agent. The properties of the algorithm are not affected. 
	
%
%
%
%
%
	
		\begin{algorithm}[h!]
		  \caption{GAL --- algorithm for envy-free assignment of indivisible objects to two agents}
		  \label{algo:GAL}
		\renewcommand{\algorithmicrequire}{\wordbox[l]{\textbf{Input}:}{\textbf{Output}:}} 
		 \renewcommand{\algorithmicensure}{\wordbox[l]{\textbf{Output}:}{\textbf{Output}:}}
		\algsetup{linenodelimiter=\,}
		\begin{algorithmic}
			\REQUIRE $((\pref_1,\pref_2), O)$
			\ENSURE EF assignment $p$
		\end{algorithmic}
		  \begin{algorithmic}[1] 
			  \normalfont
			  \STATE Construct a linear order $>_1$ for agent $1$: for all $i,j\in \{1,\ldots, m\}$, $o_i>_1 o_j$ if $o_i\succ_1 o_j$; $o_j>_1 o_i$ if $o_i\succ_2 o_j$ and $o_i\sim_1 o_j$; $o_i>_1 o_j$ if $o_i\sim_2 o_j$ and $o_i\sim_1 o_j$ and $i<j$
			  \STATE Construct the linear order $>_2$ for agent $2$: for all $i,j\in \{1,\ldots, m\}$, $o_i>_2 o_j$ if $o_i\succ_2 o_j$;
$o_j>_2 o_i$ if $o_i\succ_1 o_j$ and $o_i\sim_2 o_j$;
 $o_j>_2 o_i$ if $o_i\sim_2 o_j$ and $o_i\sim_1 o_j$ and $i<j$.

\STATE $O'\longleftarrow O$
\STATE $p(1)\longleftarrow \emptyset$; $p(2)\longleftarrow \emptyset$
	\STATE $C \longleftarrow \emptyset$ 
\STATE round number  $t\longleftarrow 0$

			\WHILE{$O'\neq \emptyset$}
			\STATE $t\longleftarrow t+1$
			\IF{$|O'|=1$}\label{step:O=1}
			\STATE \label{stepb:O=1} $C\longleftarrow C \cup O'$	
			\ELSIF{$\max_{>_1}(O')\neq \max_{>_2}(O')$}
			\STATE $p(1)\longleftarrow p(1)\cup \max_{>_1}(O')$
			\STATE $p(2)\longleftarrow p(2)\cup \max_{>_2}(O')$
			\STATE $O'\longleftarrow O'\setminus \{\max_{>_1}(O'),\max_{>_2}(O')\}$
						\ELSIF{$\max_{>_1}(O')=\max_{>_2}(O')$}
						\STATE $o^*\longleftarrow \max_{>_1}(O')$ (or $\max_{>_2}(O')$)
						\STATE $O'\longleftarrow O'\setminus \{o^*\}$
						\IF{$(p(1)\cup \{o^*\},p(2)\cup 
						 \{\max_{>_2}(O')\})$ is EF w.r.t $\pref$}\label{step:agent1}
						\STATE $p(1)\longleftarrow p(1)\cup \{o^*\}$
						\STATE $p(2)\longleftarrow p(2)\cup  \{\max_{>_2}(O')\}$
						\STATE $O' \longleftarrow O'\setminus \{\max_{>_2}(O')\}$
						\ELSIF{$(p(1)\cup \{\max_{>_1}(O')\}, p(2)\cup \{o^*\})$ is EF w.r.t $\pref$}\label{step:agent2}
						\STATE $p(2)\longleftarrow p(2)\cup \{o^*\}$
						\STATE $p(1)\longleftarrow p(1)\cup  \{\max_{>_1}(O')\}$
						\STATE $O' \longleftarrow O'\setminus \{\{\max_{>_1}(O')\}$
					\ELSE
					\STATE $C\longleftarrow C \cup \{o^*\}$	
						\ENDIF
			\ENDIF
			\ENDWHILE
		\RETURN $(p(1),p(2))$		
		 \end{algorithmic}
		\end{algorithm}

		
			First observe that for strict preferences, GAL is equivalent to the simplified AL method. The reason is that for strict preferences, there exists a unique priority order irrespective of any lexicographical tie-breaking order. We present a couple of examples to illustrate how GAL works. The contested pile is empty in one example and non-empty in another.
	
	\begin{example}
	\begin{align*}
		\pref_1:&\quad \{o_1,o_2,o_3\}, \{o_4,o_5,o_6\}\\
		\pref_2:&\quad \{o_2,o_3,o_4\},\{o_6\},\{o_1,o_5\}
	\end{align*}
	\begin{align*}
		>_1:&\quad o_1,o_2,o_3,o_5,o_6,o_4\\
		>_2:&\quad o_4,o_3,o_2,o_6,o_5,o_1
	\end{align*}
	\begin{enumerate}
		\item Round $1$: $p(1)=\{o_1\}$, $p(2)=\{o_4\}$, $C=\emptyset$;
			\item Round $2$: $p(1)=\{o_1,o_2\}$, $p(2)=\{o_4,o_3\}$, $C=\emptyset$;
												\item Round $3$: $p(1)=\{o_1,o_2,o_5\}$, $p(2)=\{o_4,o_3,o_6\}$, $C=\emptyset$. 
	\end{enumerate}
		\end{example}

		%

			\begin{example}
			\begin{align*}
				\pref_1:&\quad \{o_7\},\{o_1,o_2,o_3\}, \{o_4,o_5,o_6\}\\
				\pref_2:&\quad \{o_7\},\{o_1\},\{o_3\},\{o_4,o_5\},\{o_2,o_6\}
			\end{align*}
			\begin{align*}
				>_1:&\quad o_7,o_2,o_3,o_1,o_6,o_4,o_5\\
				>_2:&\quad o_7,o_1,o_3,o_5,o_4,o_6,o_2
			\end{align*}
		
			\begin{enumerate}
				\item Round $1$: $p(1)=\emptyset$, $p(2)=\emptyset$, $C=\{o_7\}$;
				\item Round $2$: $p(1)=\{o_2\}$, $p(2)=\{o_1\}$, $C=\{o_7\}$;
							\item Round $3$: $p(1)=\{o_2,o_3\}$, $p(2)=\{o_1,o_5\}$, $C=\{o_7\}$;
							\item Round $4$: $p(1)=\{o_2,o_3,o_4\}$, $p(2)=\{o_1,o_5,o_6\}$, $C=\{o_7\}$.			
			\end{enumerate}
				\end{example}


			\begin{proposition}
				GAL runs in $O(m^2)$ time and is deterministic. 
				\end{proposition}
				\begin{proof}
In each round, either one object each is allocated to the agents or one contested object is sent to $C$. If each agent has a different highest priority unallocated object, then the allocation takes constant time. Otherwise, the agents have the same highest priority contested object $o^*$. In this case, we need to make at most two checks for whether there exists an EF partial assignment that allocated $o^*$ to one of the agents. In either of these checks, we simply need to verify whether the given partial assignment is EF or not which takes time $O(m)$ according to Lemma~\ref{lemma:efcheck}. Thus, GAL takes time $O(m^2)$.
\qed
\end{proof}
%

					\begin{proposition}\label{prop:maximal}
						GAL returns a maximal EF assignment.
						\end{proposition}
						\begin{proof}
								The GAL outcome is EF. This follows from the way the partial assignments are constructed so that EF is maintained. If $\max_{>_1}(O')= \max_{>_2}(O')$, then the partial assignment is only modified after checking that the modification still satisfies EF. If $\max_{>_1}(O')\neq \max_{>_1}(O')$, then each agent is given a most preferred unallocated object from $O'$. Since the partial allocation $p$ is EF, and for each $o\in p(i)$, $o\pref_i \max_{>_i}(O') \pref_i \max_{>_{-i}}(O')$, it follows that the allocation which gives $p(i)\cup \{\max_{>_i}(O')\}$ to each $i\in \{1,2\}$ is EF.

We now show that the outcome is a maximal EF assignment. Assume for contradiction that GAL's outcome $p$ is not maximal EF. This means that for some object $o\in C$ there exists an assignment $q$ that matches the objects matched by $p$ as well as $o$ and possible other objects. Consider the object $o$ that is the first object to be placed in the contested pile $C$ and consider the stage in Algorithm~\ref{algo:GAL} where $o$ was sent to $C$. If $o$ was given to agent $-i$, then agent $i$ was given the next highest priority object $o'$ according to $>_i$ which still leads to infeasibility of EF.  Clearly $o\succ_i o'$ or else the partial assignment $p$ at the stage wouldn't fail EF. For every other unallocated object $o''$ in $O'$ (that has not in the contested pile) at that stage, it holds that $o'\succsim_i o''$. Hence no object $o''$ can be given to agent $i$ while $o$ is given to $-i$ so that $p$ is still $EF$. By the same argument, every subsequent object that is placed in $C$ cannot be allocated to one of the agents without causing the other agent to be envious. 
							\qed
							\end{proof}

Next, we show that if there exists a complete EF assignment, then GAL returns a complete EF assignment.
\footnote{The argument in Theorem 3 of \cite{BKK14a} only shows that for strict preferences, AL finds maximally EF assignment. It does not show that for strict preferences, AL efficiently computes a complete EF assignment if a complete EF assignment exists.}
For the proposition, we require the following lemma which follows from \citep[Theorem 4(i), ][]{AGMW14b}.

	\begin{lemma}\label{lemma:prop}
		For the case of two agents, any partial assignment $p$ if EF iff for each $o\in p(i)\cup p(-i)$,
\[|p(i)\cap \{o'\midd o'\ \pref_i o\}|\geq |\{o'\in p(i)\cup p(-i)\midd o'\pref_i o\}|/2.\]
		\end{lemma} 

		\begin{proposition}\label{prop-complete}
	If there exists a complete EF assignment, then GAL returns a complete EF assignment.
			\end{proposition}
			\begin{proof}
				Assume for contradiction that there exists a complete EF assignment but GAL does not return a complete EF assignment. Then there exists at least one object in the contested pile. Let us consider the first object $o$ that is placed in the contested pile. When $o$ is placed in the contested pile, let the partial EF allocation be $p$. Let the next priority available unallocated objects of $i$ and $-i$ be $o^i$ and $o^{-i}$ respectively where it could be possible that $o^i=o^{-i}$. Since $o$ is placed in the contested pile, this means that the assignment which gives  $p(i)\cup \{o\}$ to $i$ allocation and  $p(-i)\cup \{o^{-i}\}$ is not EF. Similarly, the assignment which gives  $p(i)\cup \{o^{i}\}$ to $i$ allocation and  $p(-i)\cup \{o\}$ is not EF. This implies that $o\succ_i o^i$ and $o\succ_{-i} o^i$.
	Let the rank $o$ in agent $i$'s priority list be $r_i(o)$ and the rank $o$ in agent $i$'s list be $r_{-i}(o)$. Now consider the objects in $p(i)$. All objects in $p(i)$ are have a better rank than $o$ for agent $i$. Secondly, in agent $i$'s priority list, if an object is not allocated to $i$, it is allocated to agent $-i$. Now agent $i$'s allocation $p(i)$ is such that if $o$ is given agent $-i$ and $o^{i}$ to agent $i$, the assignment is not EF. By Lemma~\ref{lemma:prop}, this means that
	$|p(i)|< |\{o'\in p(i)\cup p(-i)\midd o'\pref_i o\}|/2$. The assignment which allocates $o$ in addition one of the agents in addition to the partial assignment $p$ if not EF even if agent $i$ got his $|p(i)|$ most preferred objects. This means that there does not exist a complete EF assignment.\qed \end{proof}

				
				Next we show that the GAL outcome is LPO. Unlike in \citep{BKK14a}, we cannot use the characterization of \citet{BrKi05a} that if agents have \emph{strict} preferences, any assignment as a result of sequential allocation is Pareto optimal. Hence we need a lemma.
				
Let $(N,O,\pref)$ be an assignment problem and $p$ be a discrete assignment.
				We will create an auxiliary assignment problem and assignment where each agent is allocated exactly one object~\citep[see \eg ][]{AGMW14b}.
				The \emph{clones} of an agent $i\in N$ are the agents in $N_i'=\{i_o \midd o \in O \text{ and } o\in p(i)\}$.
				The \emph{cloned assignment problem} of $(N,O,\pref)$ is $(N',O,\pref')$ such that 
				$N'= \bigcup_{i\in N} N_i'$.
				and for each $i_o\in N'$, $\pref'_{i_o}=\pref_i$. The \emph{cloned assignment} of $p$ is the discrete assignment $p'$ in which $o\in p'(i_o)$ if $o\in p(i)$ and $o\notin p'(i_o)$ otherwise. A cloned assignment can easily be transformed back into the original assignment where each agent $i\in N$ is allocated all the objects assigned by $p'$ to the clones of $i$.

				%
				%


%
%

					\begin{lemma}\label{lemma:cycle2}
						An assignment for two agents is LPO iff there exist no objects $o,o'$ such that $o$ is allocated to $i$, $o'$ is allocated to $-i$, $o'\succ_i o$ and $o\succsim_{-i} o'$.
						\end{lemma}
												\begin{proof}
													By \citep[Lemma 2, ][]{AGMW14b}, 
													an assignment is Pareto optimal if and only if its cloned assignment is Pareto optimal for the cloned assignment problem. Hence, we can restrict our attention to the cloned assignment and the cloned assignment setting.
													If the cloned assignment is Pareto optimal, the original assignment is Pareto optimal. 
													If the cloned assignment is \emph{not} Pareto optimal, then there exists a `trading cycle' in which each object points to its owner, each cloned agent in the cycle points to an object that is at least as preferred as its own object and at least one agent in the cycle points to a strictly more preferred object than the one it owns~\citep{AzKe12a}.

													Firstly, we claim that there exists no trading cycle consisting only of clones of one agent. Assume for contradiction that there exist a trading cycle consisting of only of clones of the same agent. Then there exists at least one object that is minimally preferred. The agent who points to this object also owns a minimally preferred object. Hence each agent owns a minimally preferred object and thus the cycle is not Pareto improving. 

						We now show that, if there exists a trading cycle, then there exists one which alternates between clones of the two agents. Consider any cycle which has the following path consisting of multiple clones of the same agent in succession: 
						\[o_{c_1}\rightarrow i_{c_1}\rightarrow o_{c_2}\rightarrow i_{c_2} \cdots \rightarrow o_{c_k} \rightarrow i_{c_k}\rightarrow o_{c_{k+1}} {-i}_{c_{k+1}}.\]
						Since clones of each agent $i$ have the identical preference, $i_{c_1}$ also points directly to $ o_{c_{k+1}}$. Hence, we know that there is also a path 
						\[o_{c_1}\rightarrow i_{c_1} \rightarrow o_{c_{k+1}} \rightarrow {-i}_{c_{k+1}} .\]

						We now show that if there exists a trading cycle which alternates between clones of the two agents, then there exists one with exactly one clone of each agent. By the definition of trading cycle, at least one agent points to a strictly more preferred over the object he owns. 
Assume that a clone of  agent $i$ gets a strictly more preferred object in the trading cycle. Let such a clone be $i_j$ that points to $o^*$.
Consider the clone $i_1$ of agent $i$ who has the least preferred object among all clones of $i$. We can assume without loss of generality that $i_1$ points to a strictly more preferred object that the one he owns. If this were not the case, then we know that $i_j$ has a trading path to $i_1$ and $i_1$ also strictly prefers $o^*$ over the object he owns. This means that there is trading cycle in which $i_1$ points to a strictly more preferred object owned by a clone of $-i$.
Hence,  without loss of generality let the agents in the trading cycle have the following sequence where $i_1$ points to and strictly prefers the object of $-i_1$ over his own object: 
\[i_1, -i_1, i_2, -i_2, \ldots, i_k, -i_k, i_1.\]
If clone $i_2$ is indifferent between his object and the object owned by $i_1$, then this means he strictly prefers $-i_1$'s object over his own object. But this means that $i_2$ and $-i_1$ weakly prefer each other's objects over their own object and $i_2$ strictly prefers $-i_1$'s object which means we have already 
shown that there exist $o,o'$ such that $o$ is allocated to $i$, $o'$ is allocated to $-i$, $o'\succ_i o$ and $o\succsim_{-i} o'$. 
Suppose for contradiction that $i_2$ has a strictly more preferred object than the object owned by $i_1$. Since $i_1$ has the least preferred object among all clones of $i$, it points to any object that $i_2$ points to. Since $i_2$ points to the object of $-i_2$, this means that $i_1$ strictly prefers the object of $-i_2$ over his own object. By the same argument, $i_1$ strictly prefers each object owned by the clones of $-i$ in the trading cycle. Since at least one clone of $-i$ points to the object of $i_1$, we have shown that there exist $o,o'$ such that $o$ is allocated to $i$, $o'$ is allocated to $-i$, $o'\succ_i o$ and $o\succsim_{-i} o'$.
%
%
%
%
%
%
													\end{proof}

			We use Lemma~\ref{lemma:cycle2} to obtain the following proposition. 
			
			\begin{proposition}\label{prop:lpo}
The GAL outcome is LPO. 
				\end{proposition}
		\begin{proof}
	Let us constrain ourselves to the set of objects $O'\subseteq O$ that are allocated to agents $1$ and $2$. Now let $(N',O',\pref')$ be the cloned assignment problem. Then assignment $p$ for objects in $O'$ is PO iff the corresponding assignment is PO for $(N',O',\pref')$.
	Now assume that the GAL outcome is not LPO. Then the assignment with respect to $O'$ is not PO. By Lemma~\ref{lemma:cycle2}, there exists $i\in \{1,2\}$ such that $i$ gets $o$ in some round $t$, $o'\succ_i o$ where $o'$ was allocated to $-i$ and $o\succsim_{-i} o'$. This means that $o'$ was allocated to $i$ in round $t'\leq t$. 
Now if  $o\succ_{-i} o'$, then $o$ would be a higher priority object for $-i$ so that it would not have gone for $o'$ before $o$.	Then it must be that $o\sim_{-i} o'$. But, if $o\sim_{-i} o'$, then $o$ would again be a higher priority object for $-i$ so that it would not have gone for $o'$ before $o$. Hence a contradiction.
			\qed
			\end{proof}


		In Proposition~\ref{prop:lpo}, we showed that there exists no other (not necessarily EF) assignment that uses the same objects as the GAL outcome and is Pareto improvement over the GAL outcome. 
		Next we show that there exists no other EF assignment that may use any objects and is a Pareto improvement over the GAL outcome.


		\begin{proposition}\label{prop:SD-dominate}
	GAL returns an assignment such that there exists no other assignment that Pareto dominates it and is envy-free. 
			\end{proposition}
\begin{proof}
	Assume for contradiction that GAL's outcome $p$ is SD-dominated by 
another EF assignment $q$ such that $q(i)\pref_i^{SD} p(i)$ for both $i$ and 	$q(i)\succ_i^{SD} p(i)$ for at least one $i$.
				We now proceed in rounds where in each round we check the highest priority allocated object of each of the two agents that have not been checked. We check the partial assignments $p^t$ and $q^t$ in each round $t$ to see whether $q^t(i)\succ_i^{SD} p^t(i)$. Let us assume that $q^t(i)\succ_i^{SD} p^t(i)$ and $q^t(-i)\pref_{-i}^{SD} p^t(-i)$ for the smallest possible $t$ . If both $q^t(i)\succ_i^{SD} p^t(i)$ and $q^t(-i)\succ_{-i}^{SD} p^t(-i)$, 
then it means that in $q$ both get higher priority objects than $p$ in that round.				
				This is a contradiction as GAL would allocated these higher priority objects to the agents. Now assume that
$q^t(i)\succ_i^{SD} p^t(i)$ and $q^t(-i)\sim_{-i}^{SD} p^t(-i)$. This means that agent $-i$ gets an equally preferred object and the other agent $i$ gets a higher priority object. But this is again a contradiction, because GAL would have allocated the more preferred object to $i$ in that round.		\qed	
\end{proof}


		%
		%
		%
		%

		Note that for the case of two agents, \citet{AGMW14a} presented a polynomial time algorithm to check whether a complete SD-envy-free assignment exists or not. In order to compute a maximal SD envy-free assignment, one can consider different subsets  $O'\subset O$ and check whether a complete SD-envy-free assignment exists or not for $O'$. However this approach would require checking exponential number of subsets.

		We have already shown that GAL satisfies the desirable properties of AL on a more general domain. Next, we show that under strict preferences GAL returns an assignment that is a possible outcome of AL. In this sense, GAL is a `proper' generalization of AL.  
		
				\begin{proposition}
					For strict preferences, GAL returns an  AL outcome.
			\end{proposition}
					\begin{proof}
						For strict preferences, there exists a unique priority order irrespective of any lexicographical tie-breaking order. We show that under strict preferences, GAL and AL handle all the cases in an equivalent manner.

		Let us compare the formal definition of AL \citep[Page 133-134][]{BKK14a} with the pseudocode of GAL.
		In AL, in stage $t$, the direction ``If one unallocated item remains, place it in CP and stop'' is equivalent to Steps \ref{step:O=1} and \ref{stepb:O=1} of Algorithm~\ref{algo:GAL}. 

		In AL, in stage $t$, the direction ``If no unallocated items remain, stop.'' is equivalent to the stopping condition in the while loop of Algorithm~\ref{algo:GAL}.

		If both agents have different most preferred (equivalent to highest priority since the preferences are strict) unallocated objects, then both GAL and AL behave in the same manner and give the most preferred objects to the agents. For AL this direction is specified in the last sentence of the stage $t~1)$.

		Finally, both algorithms have a check for when both agents have the same most preferred objects with this check being in step  $t~2)$ in the specification of the AL method. In AL, the most preferred available contested object $i$  is tentatively given to the one of the agents.
		In the specification of Algorithm~\ref{algo:GAL}, the most preferred available object is also tentatively given to one of the agents. 
		Since, in Algorithm~\ref{algo:GAL}, this object is referred to as $o^*$ so we will refer it as $o^*$ for both algorithms. 
		Let us say agent who gets it is agent $-i$. The other agent $i$ is tentatively given the next most preferred object that is not yet allocated. In the description of AL, $i$ could be given an even less preferred unallocated object but in at least one instantiations of AL, $i$ is tentatively given the next most preferred object that is still available. 
According to AL, such a tentative assignment is \emph{feasible} as long as the number of objects assigned to $-i$ including $o^*$ or put in the contested pile (``unassigned'') that $i$ prefers to the next most preferred unallocated object is at most $t$. This means that for the tentative assignment $p$, 
$|\{o'\in p(i)\midd o'\succ_i o^* \}|\geq  |\{o'\in p(i)\cup p(-i)\midd o'\succ_i o^* \}|/2$.
Since agent $i$'s allocation from the previous round consists of objects strictly preferred over $o^*$, this means that $i$ is not envious of $-i$ in $p$ as long as $i$ was not envious of $-i$ in the previous round. 
Thus in both algorithms, the tentative assignment in which the contested object is given to agent $-i$ and the next most preferred unallocated object is given to agent $i$ is made permanent if the modification does not cause envy. Hence the feasibility check in the case of AL is equivalent to checking whether the tentative new assignment is EF. If the tentative assignment is not EF for $o^*$ given to either of the two agents, then GAL puts $o^*$ in the contested pile. Similarly, AL puts the object in the contested pile (Step $t~5)$).
		\qed
						\end{proof}

		\section{Discussion}
		
		In this paper, we presented GAL that is a generalization of the AL method of \citet{BKK14a} for the fair allocation of indivisible objects among two agents. A crucial advantage of extending AL to GAL is for the case in which agents have identical preferences. If agents have strict and identical preferences, then AL assigns all the objects to the contested pile. However if the preferences are really coarse, such as when all objects are equally preferred, then GAL assigns $\floor{m/2}$ to each agent.

		GAL can also be used as an algorithm to solve previously studied problems within fair division:

		\begin{theorem}
			There exists a $O(m^2)$ algorithm to check whether there exists a complete assignment that is EF.
			\end{theorem}
			\begin{proof}
				By Proposition~\ref{prop-complete}, if there exists a complete EF assignment, GAL returns such an assignment.
				\qed
				\end{proof}

				Previous algorithms to solve this problem take time $O(m^3)$ and require solving network flow or maximum matching problems~\citep{BEL10a,AGMW14a}.

			GAL is specifically designed for the case of two agents. This raises the question whether 
			GAL can be generalized to the case of arbitrary number of agents. 		
			
				 \begin{theorem}
					 Assume there exists an algorithm $\mathcal{A}$ that returns a maximal envy-free assignment that is complete if a complete envy-free assignment exists. Then $\mathcal{A}$ does not take polynomial time assuming $P\neq NP$.
		\end{theorem}
						\begin{proof}
							\citet{BEL10a} proved that checking whether there exists a complete EF assignment is NP-complete for strict preferences. \citet{AGMW14a} proved that checking whether there exists a complete EF assignment is NP-complete for dichotomous preferences.
If $\mathcal{A}$ is polynomial-time, then it can be used to compute a maximal EF assignment. If the assignment is complete, we know that there exists a complete EF assignment. If the assignment is not complete, we know by Proposition~\ref{prop-complete} that there does not exist a complete EF assignment. Hence a polynomial-time algorithm to compute a maximal EF  assignment can solve an NP-complete problem in polynomial time. 					\qed
				\end{proof}

		%
		%

		GAL can also be seen as a discrete version of the \emph{probabilistic serial (PS) algorithm}~\citep{BoMo01a,KaSe06a} that is used to compute a fractional assignment. PS is SD-efficient and SD-envy-free. In other words, PS returns a maximal fractional assignment that is both SD-efficient and SD-envy-free.
In the randomized setting, there is always a complete assignment that satisfies both properties.	Similarly, a GAL outcome is a maximal discrete assignment that is both SD-efficient and SD-envy-free. 
If we restrict ourselves to discrete assignments, then there may not exist a complete and envy-free assignment.

In this paper, we assumed that all objects are acceptable. The case where some objects may be unacceptable to an agent can be handled. If an object is unacceptable to both agents, it can be discarded from the outset. If an object is only acceptable to one agent, it will only be given to that agent.

It will be interesting to apply the approach of maximal EF to weaker notions of fairness~\citep{AGMW14a, BEL10a}. Finally, extending GAL to the case of constant number of agents is left as future work.

\subsubsection*{Acknowledgments}
The author thanks Steven Brams for sharing the paper on the AL method with him. He also appreciates Sajid Aziz, Steven Brams and Christian Klamler for their useful feedback.
NICTA is funded by the Australian Government through the Department of Communications and the Australian Research Council through the ICT Centre of Excellence Program. 
		

\bibliographystyle{plainnat}

\end{document}